\documentclass[letterpaper, 10 pt, conference]{ieeeconf}
\IEEEoverridecommandlockouts                              
\renewcommand{\baselinestretch}{0.97}
\usepackage[usenames,dvipsnames]{xcolor}
\usepackage{amsmath,amsthm}
\usepackage{bm}
\usepackage{amssymb}
\usepackage{mathabx}
\usepackage{dsfont}
\usepackage{graphicx}
\usepackage{epstopdf}
\usepackage[noadjust]{cite}
\usepackage{etoolbox}
\usepackage{filecontents}
\usepackage{tikz}
\usepackage[ruled,vlined]{algorithm2e}
\usepackage{subcaption}

\newtheorem{theorem}{Theorem}
\newtheorem{remark}{Remark}
\newtheorem{lemma}{Lemma}
\newtheorem{definition}{Definition}
\newtheorem{proposition}{Proposition}

\title{\LARGE \bf
Dealing with Unknown Unknowns: Identification and Selection of Minimal Sensing for Fractional Dynamics with Unknown Inputs
}

\author{Gaurav~Gupta$^{\dag}$ $\quad$ S\'ergio~Pequito$^{\ddagger}$  $\quad$  Paul~Bogdan$^{\dag}$
\thanks{
$^{\dag}$Ming Hsieh Department of Electrical Engineering, University of Southern California, Los Angeles, CA, USA {\tt\small \{ggaurav,pbogdan\}@usc.edu}}
\thanks{
$^{\ddagger}$Department of Industrial and Systems Engineering, Rensselaer Polytechnic Institute, Troy, NY, USA
{\tt\small goncas@rpi.edu}}
}

\begin{document}

\maketitle
\thispagestyle{empty}
\pagestyle{empty}

%
%
\begin{abstract}
This paper focuses on analysis and design of time-varying complex networks having  fractional order dynamics. These systems are key in modeling the complex dynamical processes arising in several natural and man made systems. Notably, examples include neurophysiological signals such as electroencephalogram (EEG) that captures the variation in potential fields, and blood oxygenation level dependent (BOLD) signal, which serves as a proxy for neuronal activity. Notwithstanding, the complex networks originated by locally measuring EEG and BOLD are often treated as isolated networks and do not capture the dependency from external stimuli, e.g., originated in subcortical structures such as the thalamus and the brain stem. Therefore, we propose a paradigm-shift towards the analysis of such complex networks under unknown unknowns (i.e., excitations). Consequently, the  main contributions of the present paper are threefold: (\emph{i}) we present an alternating scheme that enables to determine the best estimate of the model parameters and unknown stimuli; (\emph{ii}) we provide necessary and sufficient conditions to ensure that it is possible to retrieve the state and unknown stimuli; and (\emph{iii}) upon these conditions we determine a small subset of variables that need to be measured to ensure that both state and input can be recovered, while establishing sub-optimality guarantees with respect to the smallest possible subset. Finally, we present several pedagogical examples of the main results using real data collected from an EEG wearable device.

%
\end{abstract}
%

%
%
\section{Introduction}
\label{sec:intro}
A time-varying complex network with fractional dynamics has node activities over time influenced by its own history and also by activities of the other nodes. The systems with such attributes enables the modeling of coupled non-stationary processes that exhibit long-term memory~\cite{Moon,Lundstrom,Werner,Turcott,Thurner}. A multitude of examples can be found in several application domains such as biology, and engineering, e.g. bacteria dynamics \cite{Arafa2013,rihan2014fractional,Koorehdavoudi} and swarm robotics \cite{Couceiro2012,COUCEIRO201636}, respectively. Nonetheless, their applicability becomes even more useful in the context of heterogeneous networks that interact geographically and across different time-scales, as it often occurs in the context of cyber-physical systems (CPS).


In the context of the present manuscript, we are motivated by the recent success of fractional order dynamics in modeling spatiotemporal properties of physiological signals, such as electrocardiogram (ECG), electromyogram (EMG), electroencephalogram (EEG) and blood oxygenation level dependent (BOLD) just to mention a few~\cite{magin2012fractional,baleanu2011fractional}. Despite these modeling capabilities, there is one main limitation that continues to elude scientists and engineers alike. Specifically, complex networks such as the brain, whose nodes will dynamically evolve using fractal order dynamics, are often observed locally. Meaning that some of the dynamics assessed by the models are not only due to the local interaction, but might be constrained by unknown sources, i.e., stimuli that are external to the considered network. Consequently, we propose a model that enables us to account for the existence of such unknown stimuli, and determine the model that best captures the local dynamics under such stimuli. Observe that this enhances the analysis of these systems once we have an additional feature (i.e., the stimuli) that can be the main driver of a possible abnormal behavior of the network~\cite{norden2002role}.

In the context of neuroscience and medical applications, the ability to determine unknown inputs is crucial in the retrieval of the model under the so-called input artifacts~\cite{delorme2007enhanced} captured in the EEG readings due to a pulse that causes a higher variation of the potential than the baseline. Alternatively, the existence of a model to cope with the presence of unknown inputs enable the modeling of stimuli that are originated in the subcortical structures of the brain that are often neglected in the current EEG and functional magnetic resonance imaging (fMRI) that leverages the BOLD signals~\cite{shah2004neural}. Thus, it is imperative to develop such framework, as well as tools that enable us to perform the analysis and design of the systems modeled by such setup, which we introduce in the current paper.

To the best of the authors' knowledge, such framework has not been advanced in the context of discrete-time fractional order dynamics. In the domain of continuous-time fractional dynamics, works like \cite{KONG20175503,Doye} exist for the design of observer in the presence of unknown inputs. The closest work for the discrete-time case is \cite{xue1}, which does not consider the case of unknown inputs.  Nonetheless, the usefulness of accounting for unknown inputs in the context of linear time invariant (LTI) systems is an old topic \cite{charandabi2014novel,sundaram2006optimal,yang2007sequential,park1988closed,hsieh2010optimality,sundaram2007delayed}. Specifically, the closest papers to the results proposed in this paper are as follows: observer with unknown inputs~\cite{sundaram2006optimal}, delayed systems with unknown inputs \cite{sundaram2007delayed}, estimation of unknown inputs with sparsity constraints~\cite{sefati2015linear}. Notwithstanding, LTI systems are not good approximations for fractional order dynamical systems due to their limited memory capabilities. Yet, due to the numerical approximation of the highly nonlinear infinite dimension fractional order system used, we are able to obtain finite dimension closed-form description that will enable us to derive results alike those attained in the context of LTI systems.

The  main contributions of the present paper are threefold: (\emph{i}) we present an alternating scheme that enables to determine the best estimate of the model parameters and unknown stimuli; and (\emph{ii}) we provide necessary and sufficient conditions to ensure that it is possible to retrieve the state and unknown stimuli; and (\emph{iii}) upon these conditions we determine a small subset of variables that need to be measured to ensure that both state and input can be recovered, while establishing sub-optimality guarantees with respect to the smallest possible subset.

The remaining of the paper is organized as follows. Section~\ref{sec:probForm} introduces the model considered in this paper and the main problems studied in this manuscript. Next, in Section~\ref{sec:modelEst} and~\ref{sec:sensorSel}, we present the solution to these problems. Finally, in Section~\ref{sec:experi}, we present an illustrative example of the main results using real EEG data from a wearable technology. 

%
%

\section{Problem Formulation}
\label{sec:probForm}

In this section, we first describe the time-varying complex network model having fractional order dynamical growth under unknown excitations. Next, upon this model, we propose two main problems regarding analysis and design to be addressed in the present paper.

\subsection{System Model}
\label{ssec:sysMod}
We consider a linear discrete time fractional-order dynamical model described as follows:
\begin{eqnarray}
\Delta^{\alpha}x[k+1] &=& Ax[k] + Bu[k] \nonumber\\
y[k] &=& Cx[k], 
\label{eqn:fracLlinModel}
\end{eqnarray}
\noindent where $x\in\mathbb{R}^{n}$ is the state, $u \in \mathbb{R}^{p}$ is the unknown input and $y \in \mathbb{R}^{n}$ is the output vector. Also, we can describe the system by its matrices tuple $(\alpha,A, B, C)$ of appropriate dimensions. In what follows, we assume that the input size is always strictly less than the size of state vector, i.e., $p < n$. The difference between a classic linear time-invariant and the above model is the inclusion of fractional-order derivative whose expansion and discretization \cite{andrzej} for any $i$th state $(1\leq i\leq n)$ can be written as
\begin{equation}
\Delta^{\alpha_{i}}x_{i}[k] = \sum\limits_{j=0}^{k}\psi(\alpha_{i},j) x_{i}[k-j],
\label{eqn:fracExpan}
\end{equation}
\noindent where $\alpha_{i}$ is the fractional order corresponding to the $i$th state and $\psi(\alpha_{i},j) = \frac{\Gamma(j-\alpha_{i})}{\Gamma(-\alpha_{i})\Gamma(j+1)}$ with $\Gamma(.)$ denoting the gamma function. 

Having defined the system model, the system identification i.e. estimation of model parameters from the given data is an important step. It becomes nontrivial when we have unknown inputs since one has to be able to differentiate which part of the evolution of the complex network is due to its intrinsic dynamics and what is due to the unknown inputs. Subsequently, one of the first problems we need to address is that of system identification from the data, as described next.

\subsection{Data-driven Model Estimation}

The fractional-order dynamical model takes care of long-range memory which often is the property of many physiological signals. The estimation of the spatiotemporal parameters when there are no inputs to the system was addressed in~\cite{xue1}. But as it happens, ignoring the inputs inherently assume that the system is isolated from the external surrounding.  Hence, for a model to be able to capture the system dynamics, the inclusion of unknown inputs is necessary. Therefore, the first problem that we consider is as follows. 

\textbf{\textit{Problem-1}}: Given the input coupling matrix $B$, and measurements of all states across a time horizon $[t,t+T-1]$ of length $T$, we aim to estimate the model parameters $(\alpha, A)$ and the unknown inputs $\{u[k]\}_{t}^{t+T-2}$. 

Notice that this would extend the work in \cite{xue1} to include the case of unknown inputs. In fact, we will see in Section~\ref{sec:modelEst} that the proposed solution would result in a different set of model parameters.


\subsection{Sensor Selection}

For the system model described by (\ref{eqn:fracLlinModel}), where the system parameters were determined as part of the solution to the \textit{Problem-1}, we consider that the output measurements are collected only by a subset of sensors. In numerous applications (for example physiological systems) it happens that the sensors are dedicated, i.e., each sensor capture an individual state \cite{xue}, so the measurement model can be written as
\begin{equation}
 y[k] = \mathbb{I}^{S}x[k],
\end{equation}

\noindent where $\mathbb{I}^{S}$ is the matrix constructed by selecting rows indexed by set S of the $n\times n$ identity matrix  $I_{n}$. As an example, if all sensors are selected, i.e., $S = [n]\equiv \{1,2,\hdots,n\}$, then $\mathbb{I}^{S} = I_{n}$. For selecting the best set of sensors $S$, with knowledge of the system matrices and the given observations, we would resort to the constraint of \emph{perfect observability} that is defined as follows.

\begin{definition}[Perfect Observability] A system described by~\eqref{eqn:fracLlinModel} is called perfectly observable if given the system matrices $(\alpha,A,B,C)$ and $K$ observations $y[k], 0\leq k\leq K-1$, it is possible to recover the initial state $x[0]$ and the unknown inputs $\{u[k]\}_{k=0}^{K-2}$.
\label{defn:perfObserv}
\end{definition}

Subsequently, the second problem that we consider is as follows. \textbf{\textit{Problem-2}}: Determine the  minimum number of sensors $S$ such that the system whose dynamics is captured by $(\alpha,A,B,\mathbb{I}^{S})$ is perfectly observable from the collection of $K$  measurements collected by a sub-collection of $S$ dedicated outputs, i.e.,
\begin{eqnarray}
&\min\limits_{S \subseteq [n]}\vert S\vert\qquad \nonumber\\
&\text{s.t.}\quad \text{$(\alpha,A,B,\mathbb{I}^{S})$ is perfectly observable.}
\label{opt:minSensProbWords}
\end{eqnarray}

In section~\ref{sec:sensorSel}, we will derive the mathematical formulation in terms of algebraic constraints of the perfect observability, which later be used to obtain a solution to~\eqref{opt:minSensProbWords}.

%
%

\section{Model Estimation}
\label{sec:modelEst}

We consider the problem of estimating $\alpha$, $A$ and inputs $\{u[k]\}_{t}^{t+T-2}$ from the given limited observations $y[k]$, $k = [t, t + T -1]$, which due to the dedicated nature of sensing mechanism is same as $x[k]$ and under the assumption that the input matrix $B$ is known. The realization of $B$ can be application dependent and is computed separately using experimental data. For the simplicity of notation, let us denote $z[k] = \Delta^{\alpha}x[k+1]$ with $k$ chosen appropriately. The pre-factors in the summation in (\ref{eqn:fracExpan}) grows as $\psi(\alpha_{i},j) \sim \mathcal{O}(j^{-\alpha_{i}-1})$ and, therefore, for the purpose of computational ease we have limited the summation in (\ref{eqn:fracExpan}) to the first $J$ values, where $J>0$ is sufficiently large. Therefore, $z_{i}[k]$ can be written as
\begin{equation}
z_{i}[k] = \sum\limits_{j=0}^{J-1}\psi(\alpha_{i},j)x[k+1-j],
\label{eqn:zDefn}
\end{equation}
with the assumption that $x[k], u[k] = 0$ for $k\leq t-1$. Using the above introduced notations and the model definition in (\ref{eqn:fracLlinModel}), the given observations are written as
\begin{equation}
z[k] = Ax[k] + Bu[k] + e[k],
\end{equation}
\noindent where $e \sim \mathcal{N}(0,\Sigma)$ is assumed to be  Gaussian noise independent across space and time. For simplicity, we have assumed that each noise component has same variance, i.e., $\Sigma = \sigma^2 I$. Also, let us denote the system matrices as $A = [a_{1},a_{2},\hdots,a_{n}]^{T}$ and $B = [b_{1},b_{2},\hdots,b_{n}]^{T}$. The vertical concatenated states and inputs during an arbitrary window of time as  $X_{[t-1,t+T-2]} = [x[t-1],x[t],\hdots,x[t+T-2]]^{T}$, $U_{[t-1,t+T-2]} = [u[t-1],u[t],\hdots,u[t+T-2]]^{T}$ respectively, and for any $i$th state we have $Z_{i,[t-1,t+T-2]} = [z_{i}[t-1],z_{i}[t],\hdots,z_{i}[t+T-2]]^{T}$. For the sake of brevity, we would be dropping the time horizon subscript from the above matrices as it is clear from the context.


Since the problem of joint estimation of the different parameters is highly nonlinear, we proceed as follows: (\emph{i}) we estimate the fractional order $\alpha$ using the wavelet technique described in \cite{flandrin}; and (\emph{ii}) with $\alpha$ known, the $z$ in (\ref{eqn:zDefn}) can be computed under the additional assumption that the system matrix $B$ is known. Therefore, the problem now reduces to estimate $A$ and the inputs $\{u[k]\}_{t}^{t+T-2}$. Towards this goal, we propose the following sequential optimization algorithm similar to an expectation-maximization (EM) algorithm~\cite{McLachlam}. Briefly, the EM algorithm is used for maximum likelihood estimation (MLE) of parameters subject to hidden variables. Intuitively, in our case, in Algorithm~1, we estimate $A$ in the presence of hidden variables or \textit{unknown unknowns} $\{u[k]\}_{t}^{t+T-2}$. Therefore, the `E-step' is performed to average out the effects of unknown unknowns and obtain an estimate of $u$, where due to the diversity of solutions, we control the sparsity of the inputs (using the parameter $\lambda'$). Subsequently, the `M-step' can then accomplish MLE estimation to obtain an estimate of $A$. The solution provided in \cite{xue1} can be related to the proposed technique as follows.


\begin{remark}
The solution to the system parameters $(\alpha, A)$ estimation without inputs \cite{xue1} is a special case of the EM like approach proposed in the Algorithm\,\ref{alg:EM_alg}.
\end{remark}

\begin{proof}
Upon setting $\{u[k]\}_{t}^{t+T-2}=0$ in the E-step of the Algorithm\,\ref{alg:EM_alg}, M-step would be the same at each iteration. Hence the algorithm stays at the initial point which is the solution in \cite{xue1}.
\end{proof}

%
%

\begin{algorithm}
{\small
\SetKw{KwInitialize}{Initialize}
\SetArgSty{normal}
\KwIn{ $x[k], k \in [t,t+T-1]$ and $B$}

\KwOut{$A$ and $\{u[k]\}_{t}^{t+T-2}$}

\KwInitialize{compute $\alpha$ using \cite{flandrin} and then $z[k]$}. For $l=0$, initialize $A^{(l)}$ as
\begin{equation*}
a_{i}^{(l)} = \text{arg}\min\limits_{a} \vert\vert Z_{i} - Xa\vert\vert_{2}^{2}
\end{equation*}
\Repeat{until converge}{

(i) \textbf{`E-step'}: For $k \in [t,t+T-2]$ obtain $u[k]$ as
\begin{equation*}
u[k] = \text{arg}\min\limits_{u}\vert\vert z[k] - A^{(l)}x[k] - Bu\vert\vert_{2}^{2} + \lambda'\vert\vert u\vert\vert_{1},
\end{equation*}
where $\lambda' = 2\sigma^{2}\lambda$;

(ii) \textbf{`M-step'}: \\ obtain $A^{(l+1)}= [a_{1}^{(l+1)},a_{2}^{(l+1)},\hdots,a_{n}^{(l+1)}]^{T}$ where
\begin{equation*}
a_{i}^{(l+1)} = \text{arg}\min\limits_{a} \vert\vert \tilde{Z}_{i} - Xa\vert\vert_{2}^{2},
\end{equation*}
\noindent and $\tilde{Z}_{i} = Z_{i} - Ub_{i}$\;
$l \leftarrow l + 1$\;
}
}
\caption{EM algorithm}
\label{alg:EM_alg}
\end{algorithm}
%
It is worthwhile to mention the following result regarding EM algorithm.

\begin{theorem}[\cite{dempster}]
The incomplete data (without unknown unknowns) likelihood $\mathbb{P}(z,x;A^{(l)})$ is non-decreasing after an EM iteration.
\label{thm:liklIncr}
\end{theorem}

Hence, the proposed algorithm being EM (detailed formulation in the Appendix\,\ref{appd:EMFormulation}) has non-decreasing likelihood. Additionally, we have the following result about the incomplete data likelihood.

\begin{proposition}
The incomplete data likelihood $\mathbb{P}(z,x;A^{(l)})$ is bounded at each iteration $l$.
\label{prop:liklBounded}
\end{proposition}

We can comment about the convergence of the Algorithm\,\ref{alg:EM_alg} as follows.
\begin{lemma}
The Algorithm\,\ref{alg:EM_alg} is convergent in the likelihood sense.
\end{lemma}
\begin{proof}
Using Theorem\,\ref{thm:liklIncr}, Proposition\,\ref{prop:liklBounded} and Monotone Convergence Theorem, we can claim that the likelihood $\mathbb{P}(z,x;A^{(l)})$ will converge.
\end{proof}

It should be noted that convergence in likelihood does not always guarantee convergence of the parameters. But as emphasized in \cite{wu1983}, from numerical viewpoint the convergence of parameters is not as important as convergence of the likelihood. Also the EM algorithm can converge to saddle points as exemplified in \cite{McLachlam}.
%
%

\section{Sensor Selection}
\label{sec:sensorSel}

Before defining the problem of sensor selection, we review the properties of fractional-order growth system with closed-form expressions for state vectors. Using the expansion of fractional order derivative in (\ref{eqn:fracExpan}), we can write the state evolving equation as
\begin{equation}
 x[k+1] = Ax[k] - \sum\limits_{j=1}^{k+1}D(\alpha, j)x[k+1-j] + Bu[k],
 \label{eqn:nextStateRaw}
\end{equation}
\noindent where $D(\alpha,j) = \text{diag}(\psi(\alpha_{1},j),\psi(\alpha_{2},j),\hdots,\psi(\alpha_{n},j))$. Alternatively, (\ref{eqn:nextStateRaw}) can be written as

\begin{equation}
 x[k+1] = \sum\limits_{j=0}^{k}A_{j}x[k-j] + Bu[k],
 \label{eqn:nextStateFin}
\end{equation}
\noindent where $A_{0} = A-D(\alpha,1)$ and $A_{j} = -D(\alpha,j+1)$ for $j\geq 1$. With this definition of $A_{j}$, we can define the matrices $G_{k}$  as follows~\cite{guermah}:

\begin{equation}
 G_{k} = \begin{cases}
          I_{n} & k = 0,\\
          \sum\limits_{j=0}^{k-1}A_{j}G_{k-1-j} & k\geq 1.
         \end{cases}
\label{eqn:GMatrix}
\end{equation}
Thus, we can obtain the following result. 
\begin{lemma}[\cite{guermah}]
The solution to system described by (\ref{eqn:fracLlinModel}) is given by
\begin{equation}
 x[k] = G_{k}x[0] + \sum\limits_{j=0}^{k-1}G_{k-1-j}Bu[j].
 \label{eqn:nextStateThm}
\end{equation}
\end{lemma}

\subsection{System Observability}

To achieve perfect observability, i.e., to retrieve the initial state and the unknown inputs, we need system matrices and $K$ observations. While any observation matrix $C$ is sufficient for defining the perfect observability, if we set $C = \mathbb{I}_{S}$ as introduced in Section~\ref{sec:probForm}, then $K$ and $S$ are intertwined. Simply speaking, by increasing $K$ we will have more measurements  acquired across time which could be used to compensate the number of measurements at each instance of time ruled by the set $S$ of sensors used.


Given the $K$ observations from $k = 0$ to $K-1$, we can represent the initial state $x[0]$ and the unknown inputs using~\eqref{eqn:nextStateThm} by defining the following matrices

\begin{equation}
 \Theta = [(CG_{0})^{T}, (CG_{1})^{T},\hdots,(CG_{K-1})^{T}]^{T},
 \label{eqn:ThetaConstr}
\end{equation}

and

\begin{equation}
 \Xi = \begin{bmatrix}
       0 & 0 & \hdots & 0 & 0 \\
       CG_{0}B & 0 & \hdots & 0 & 0 \\
       CG_{1}B & CG_{0}B & \hdots & 0 & 0 \\
       \vdots & \vdots & \ddots & \vdots & \vdots \\
       CG_{K-2}B & CG_{K-3}B & & CG_{0}B & 0 
      \end{bmatrix},
       \label{eqn:XiConstr}
\end{equation}

\noindent where $C$ and $B$ are the observation and input matrices from (\ref{eqn:fracLlinModel}) respectively, and $G_{k}$ is as defined in~\eqref{eqn:GMatrix}. Having $\Theta$ and $\Xi$ defined and using (\ref{eqn:fracLlinModel}), we can write the initial state and inputs in terms of the observations as
\begin{equation}
 Y = \Theta\,x[0] + \Xi\,U,
 \label{eqn:observEqn}
\end{equation}

\noindent where $Y = [y[0]^{T},y[1]^{T},\hdots,y[K-1]^{T}]^{T}$ and $U = [u[0]^{T},u[1]^{T},\hdots,u[K-1]^{T}]^{T}$.

Using (\ref{eqn:observEqn}) and the Definition\,\ref{defn:perfObserv}, a necessary and sufficient condition to attain the perfect observability is obtained as follows.

\begin{proposition}\label{prop:condPerfObs}
The system described by (\ref{eqn:fracLlinModel}) is perfectly observable after $K$ measurements if and only if $\text{rank}([\Theta\:\:\Xi]) = n + (K-1)p$.
\end{proposition}

\begin{proof}
The proof follows from re-writing equation (\ref{eqn:observEqn}) as

\begin{equation}
Y = [\Theta\:\:\Xi]\begin{bmatrix}
x[0] \\
U
\end{bmatrix},
\end{equation}
\noindent and, therefore, $x[0]$ and first $K-1$ inputs from $U$ can be recovered if and only if $\text{rank}([\Theta\:\:\Xi]) = n + (K-1)p$.
\end{proof}

Proposition~\ref{prop:condPerfObs} will be key to formulate the constraints in the sensor selection problem as detailed in the next section. 

\subsection{Sensor Selection Problem}

Given the system matrices $(\alpha, A, B)$ and first $K$ observations, the problem of sensor selection is defined as selecting the minimum number of sensors such that the system is perfectly observable. It can be mathematically written as
\begin{equation}
\min\limits_{S \subseteq [n]}\vert S\vert\qquad \text{s.t.}\quad \text{rank}([\Theta\:\:\Xi] \big\vert C = \mathbb{I}^{S}) = n+(K-1)p
\label{opt:minSensProb}
\end{equation}

\noindent where $\text{rank}([\Theta\,\,\Xi] \big\vert C = \mathbb{I}^{S})$  denotes the rank of $[\Theta\,\,\Xi]$ matrix when $\Theta$ and $\Xi$ are constructed from (\ref{eqn:ThetaConstr}) and (\ref{eqn:XiConstr}) with $C = \mathbb{I}^{S}$. An analogous problem of sensor selection with no inputs is studied in \cite{xue} and it was shown to be NP-hard; hence,~\eqref{opt:minSensProb} is at least as computationally challenging since it contains the former as a special case when $B=0$. 

Consequently, we propose a  sub-optimal solution to the above problem, while providing optimality guarantees within constant factor. For the discrete set $\Omega = [n]$, a function $f:2^{\Omega}\rightarrow \mathbb{R}$ is called submodular if for all sets $S, T \subseteq \Omega$
\begin{equation}
f(S \cup T) + f(S \cap T) \leq f(S) + f(T).
\end{equation}

Also, the marginal of an element $a$ w.r.t. set $S$ is defined as $f_{S}(a) = f(S \cup a) - f(S)$. Alternatively, a set function is referred as submodular if and only if it satisfies the diminishing returns property, i.e., $f_{S}(a) \geq f_{T}(a)$ for all $S\subseteq T\subseteq \Omega$ and $a \in \Omega\setminus T$~\cite{bach}.  The monotone submodular functions have a remarkable property of performance through greedy selection within constant factor of the optimality \cite{conforti}.

With the motivation to borrow such attributes, let us define a discrete set function $f(S)$ as, $f(S) = \text{rank}([\Theta \:\: \Xi] \big\vert C = \mathbb{I}^{S})$.

\begin{theorem}
The discrete set function $f(S) = \text{rank}([\Theta \:\: \Xi] \big\vert C = \mathbb{I}^{S})$ is submodular in $S$.
\label{thm:subMod}
\end{theorem}

%


Since $f$ is submodular, we will be using a greedy selection of sensors to maximize the rank of $[\Theta\:\:\Xi]$; in other words, greedily select sensors such that $f(S) = n + (K-1)p$. The sensor selection algorithm is described as Algorithm\,\ref{alg:greedy_alg}.

\begin{lemma}
The complexity of Algorithm\,\ref{alg:greedy_alg} with a total of $n$ sensors and $K$ length time horizon is $\mathcal{O}(n^{5}K^{3})$ i.e polynomial.
\label{lemm:polyCompl}
\end{lemma}
\begin{proof}
With $\Omega = [n]$ the computation of $f_{S_{G}}(s)$ would require at most $\mathcal{O}(n^{3}K^{3})$. The algorithm being forward greedy selection has at most $n(n+1)/2$ executions and hence the complexity of the algorithm is at most $\mathcal{O}(n^{5}K^{3})$.
\end{proof}

Therefore with Theorem\,\ref{thm:subMod} and Lemma\,\ref{lemm:polyCompl}, the Algorithm\,\ref{alg:greedy_alg} provide a sub-optimal solution with optimality guarantees within constant factor to the NP-hard problem (\ref{opt:minSensProb}) in the polynomial order complexity.

\begin{figure}
\centering
\scalebox{0.75}{
\begin{tikzpicture}[scale = 1.4]
\node[anchor=north west,inner sep=0] at (0,0) {\includegraphics*[viewport=0 0 520 500, width = 3.35in, height = 3in]{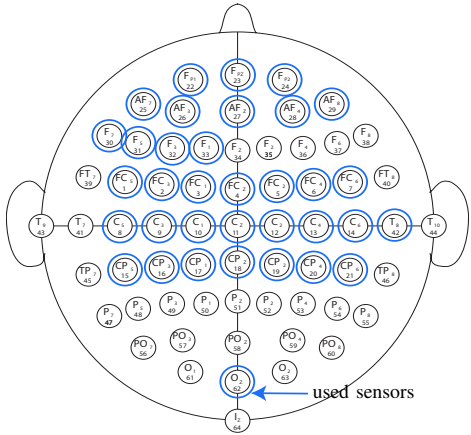}};
\node[anchor=north west] at (3.8,-4.75) {used sensors};
\end{tikzpicture}
}
\caption{Sensor distribution for the measurement of EEG. The channel labels are shown with their corresponding number.}
\label{fig:usedSensors}
\end{figure}


%
%

\begin{algorithm}
{\small
\SetKw{KwInitialize}{Initialize}
\SetArgSty{normal}
%
\KwOut{$S_{G}$}

\KwInitialize{$S_{G} = \phi$\;}

\Repeat{$f(S_{G}) = n+(K-1)p$}{

$s^{\ast} = \text{arg}\max\limits_{s \in \Omega\setminus S_{G}}f_{S_{G}}(s)$\;

$S_{G} \leftarrow S_{G} \cup s^{\ast}$\;

}
}
\caption{Greedy Sensor Selection}
\label{alg:greedy_alg}
\end{algorithm}
%

%
%
\section{Experiments}
\label{sec:experi}
We demonstrate the application of the results derived in the previous sections on physiological signals. In particular we have taken a $64$-channel EEG signal which records the brain activity of 109 subjects. The $64$-channel/electrode distribution with the corresponding labels and numbers are shown in Figure\,\ref{fig:usedSensors}. The subjects were asked to perform motor and imagery tasks. The data was collected by BCI$2000$ system with sampling rate of 160Hz \cite{schalk,goldberger}.

\begin{figure}
\centering
\begin{subfigure}[t]{0.5\textwidth}
\centering
\scalebox{0.8}{
\includegraphics*[viewport=60 310 940 650, width = 3.35in, height = 1.5in]{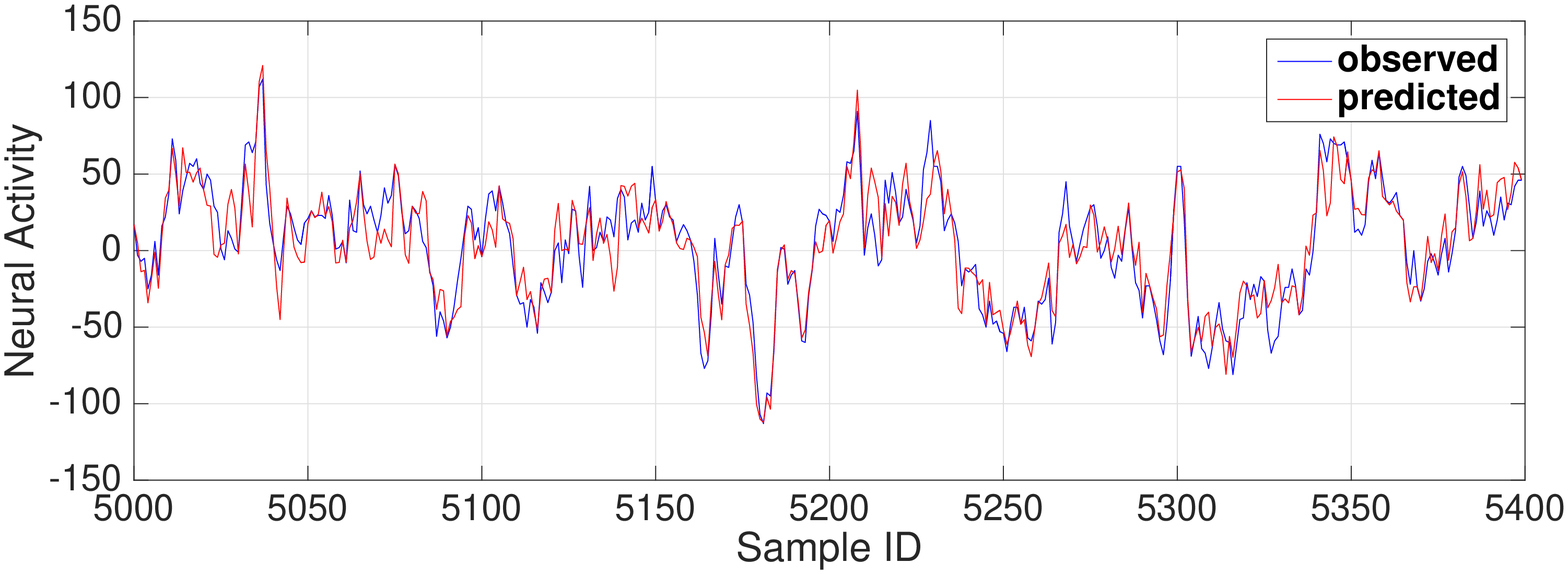}
}
\caption{}
\label{sfig:fiveStep}
\end{subfigure}
\begin{subfigure}[t]{0.5\textwidth}
\centering
\scalebox{0.8}{
\includegraphics*[viewport=60 315 880 630, width = 3.35in, height = 1.5in]{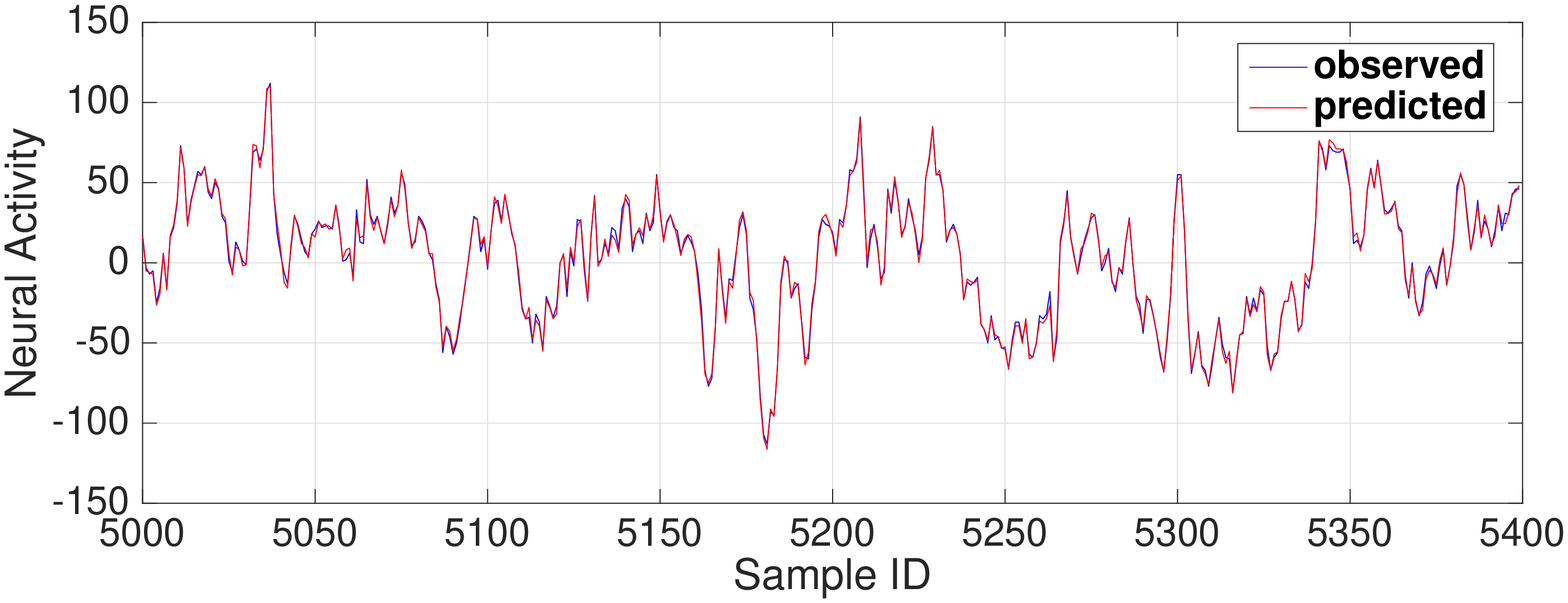}
}
\caption{}
\label{sfig:oneStep}
\end{subfigure}
\caption{Comparison of predicted EEG state for the channel $C_{1}$ using fractional-order dynamical model. The five step and one step predictions are shown in (\ref{sfig:fiveStep}) and (\ref{sfig:oneStep}) respectively.}
\label{fig:sim1}
\end{figure}

%

\subsection{Model parameters estimation}

The parameters of the system model $\alpha$ and $A$, are estimated by the application of Algorithm\,\ref{alg:EM_alg}. The performance of EM algorithm like any iterative algorithm is crucially dependent on its initial conditions. For the considered example of EEG dataset, it was observed that convergence of the algorithm is fast. Further, even a single iteration was sufficient to reach the point of local maxima of the likelihood. This shows that the choice of the initial point for EM algorithm is considerably good. The input coupling matrix $B$ can be easily determined through experiments. The values predicted by the model in comparison with actual data are shown in Figure\,\ref{fig:sim1}. The one step prediction follows very closely the actual data, but there is small mismatch in the five step prediction. The ratio of square root of mean squared error of the prediction by model with and without inputs~\cite{xue} is shown in Figure~\ref{fig:relErrComp} for total of $109$ subjects. As observed, the error ratio is less than one-third in the case when unknown inputs is considered. Therefore, the fractional-order dynamical model with unknown inputs fits the EEG data much better than the case of no inputs. In the next part, we will use these estimated parameters to compute the set of sensors for perfect observability.

\begin{figure}
\centering
\scalebox{0.75}{
\includegraphics*[viewport=30 10 600 490, width = 3.35in, height = 3in]{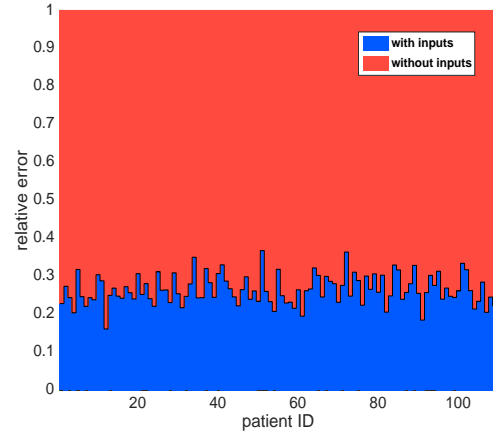}
}
\caption{Error ratio for prediction using fractional-order dynamical model with and without inputs.}
\label{fig:relErrComp}
\end{figure}

\begin{figure}
\centering
\scalebox{0.75}{
\begin{tikzpicture}[scale = 1.4]
\node[anchor=north west,inner sep=0] at (0,0) {\includegraphics*[viewport=60 20 925 690, width = 3.35in, height = 2.8in]{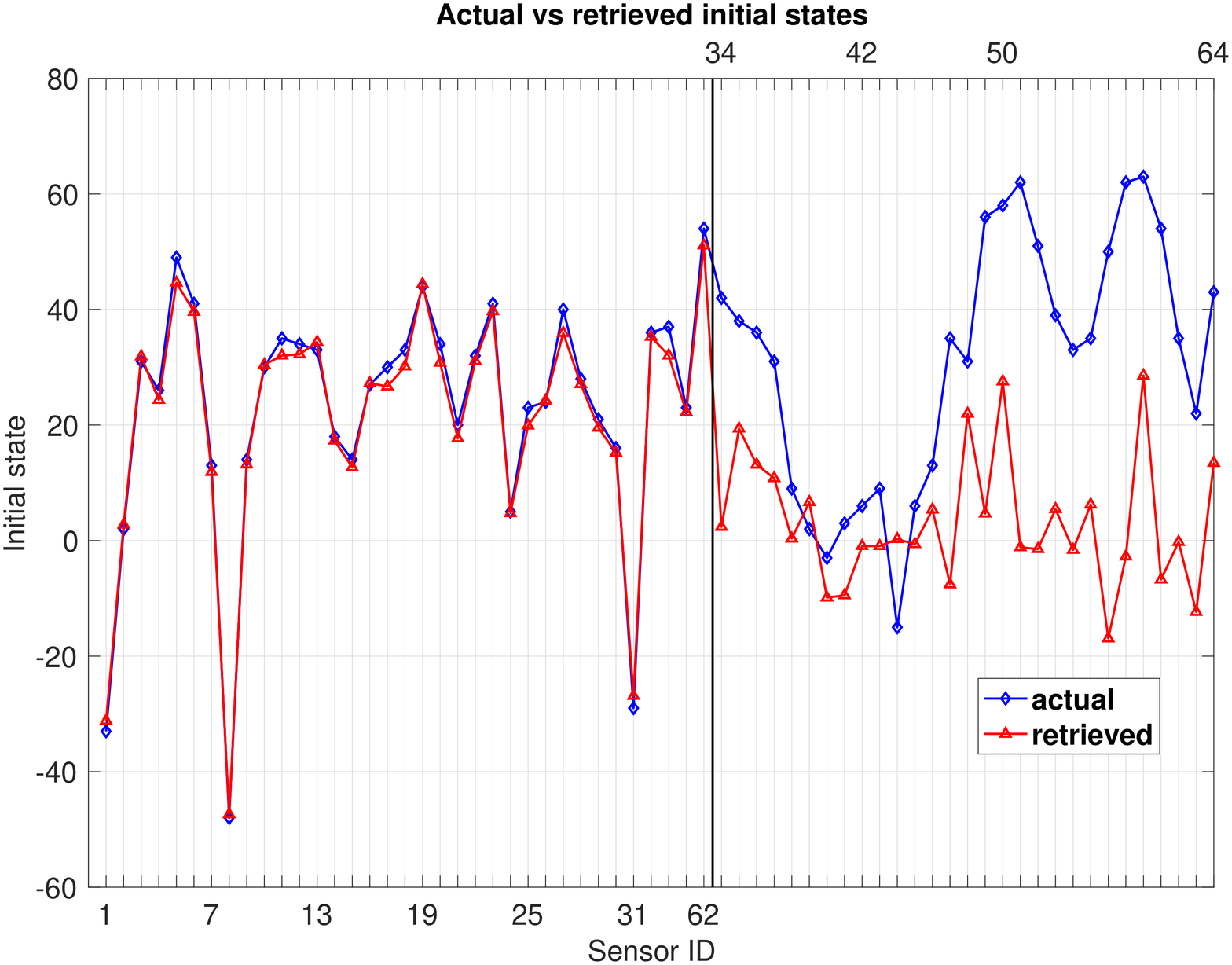}};
\draw[->, color={rgb:red,28;green,111;blue,244}, line width = 0.8pt] (3.55,-4) -- (3,-4);
\node[anchor=north west] at (1.8,-3.8) {\footnotesize{used sensors}};
\node[anchor=north west] at (1.8,-4) {\footnotesize{$(S_{G})$}};

\draw[->, color=Gray, line width = 0.8pt] (3.55,-0.9) -- (4.1,-0.9);
\node[anchor=north west] at (3.8,-0.5) {\footnotesize{unused sensors $(\Omega\setminus S_{G})$}};
\end{tikzpicture}
}
\caption{Actual vs predicted initial state using only subset of sensors $S_{G}$, selected by the Algorithm-\ref{alg:greedy_alg}.}
\label{fig:sim2}
\end{figure}

\subsection{Sensor selection}

The estimated parameters are used to construct $\Theta$ and $\Xi$ as written in (\ref{eqn:ThetaConstr}) and (\ref{eqn:XiConstr}) for the greedy sensor selection Algorithm\,\ref{alg:greedy_alg}. Upon the application of Algorithm\,\ref{alg:greedy_alg}, we found that roughly half of the sensors $(35$ out of $64)$ are sufficient enough to retrieve the initial state and unknown inputs uniquely. The selected sensors for a particular subject are marked in the Figure\,\ref{fig:usedSensors}. Due to the large size and sparsity of $[\Theta\,\,\Xi]$ matrix, some values of the initial state may blow up due to the presence of very small eigenvalues. In such cases, we can first remove the unknown inputs from (\ref{eqn:observEqn}) by multiplying both sides by $W = I - \Xi(\Xi^{T}\Xi)^{-1}\Xi^{T}$, i.e., projecting signals into the orthogonal space of $\Xi$. We can then enforce the norm constraint of $x[0]$ into the least squares estimation of $x[0]$ or, in other words, use RIDGE regression \cite{Murphy}, i.e.,
\begin{equation*}
 x[0] = \text{arg}\min\limits_{x}\vert\vert WY - W\Theta\,x\vert\vert_{2}^{2} + \epsilon\vert\vert x\vert\vert_{2}^{2},
\end{equation*}
where $\epsilon>0$.

Upon the knowledge of the initial state, the unknown inputs are recovered in the similar fashion or by enforcing sparsity constraint. Figure\,\ref{fig:sim2} shows a comparison between the actual and retrieved initial states when using the sensor set $S_{G}$ as output of the Algorithm\,\ref{alg:greedy_alg}. The retrieved initial states are close to the actual values provided they are estimated in the presence of numerical errors and sparsity.

The presented experimental results show how the proposed schemes are useful in mapping the complex dynamics of brain in the presence of unknown stimuli. The same framework can be easily applied to the study of other complex time evolving networks such as the physiological dynamics systems, for example BOLD signals, EMG, ECG etc.

%
%

\section{Conclusion}
\label{sec:concl}

In this paper, we introduced the framework of discrete fractional order dynamical systems under unknown inputs. Also, we provided tools to perform analysis and design of such systems. Specifically, for the analysis, we presented an alternating scheme that enables to determine the best estimate of the model parameters and unknown stimuli. Also,  we provided necessary and sufficient conditions to ensure that it is possible to retrieve the state and unknown stimuli. Furthermore, we enable the design of sensing capabilities of such systems, and provided a mechanism to determine a small subset of variables that need to be measured to ensure that both state and input can be recovered, while establishing sub-optimality guarantees with respect to the smallest possible subset.

Future research will focus on exploiting the structure of fractional order dynamical systems in the context of multi-case scenarios under quantitative description of the estimation quality of the states and inputs. Such extension will enable to determine the confidence in the model obtained that would permit formal claims about the mechanism underlying in the process under study. Additionally, some of the algorithms need to be improved to be deployed in real-time applications when energy and computational resources are limited. 

\section{Acknowledgment}

The authors are thankful to the reviewers for their valuable feedback. G.G. and P.B. gratefully acknowledge the support by the U.S. Army Defense Advanced Research Projects Agency (DARPA) under grant no. W911NF-17-1-0076 and DARPA Young Faculty Award under grant no. N66001-17-1- 4044.

%
%
{
\appendices
{\small
\vspace{-0.2cm}

\section{EM Formulation}
\label{appd:EMFormulation}

We present the detailed construction of EM like algorithm in this section. In our formulation, the observed (incomplete) data is $x$ and $z$ while $u$ is the hidden data, therefore the complete data would be $(z,x,u)$. Let us consider, $\Sigma = \sigma^{2}I$ and at the $l$th iteration we denote
\begin{equation*}
u[k]^{\ast} = \text{arg}\max\limits_{u}\mathbb{P}\left(u \big\vert z[k], x[k]; A^{(l)}\right).
\end{equation*}

We can enforce Laplacian prior for $u[k]$ for sparsity (any other prior could also be used) such that $\mathbb{P}(u[k]) \propto \text{exp}(-\lambda\vert\vert u[k]\vert\vert_{1})$. Therefore, $u[k]^{\ast}$ is then derived as
\begin{eqnarray*}
u[k]^{\ast} &=& \text{arg}\max\limits_{u}\text{log}\mathbb{P}\left(u \big\vert z[k], x[k]; A^{(l)}\right) \\
&=& \text{arg}\max\text{log}\mathbb{P}(u) + \text{log}\mathbb{P}\left( z[k], x[k]\big\vert u;A^{(l)}\right)\\
&=& \text{arg}\max-\frac{1}{2\sigma^{2}}\vert\vert z[k] - A^{(l)}x[k] - Bu\vert\vert_{2}^{2} -\lambda\vert\vert u\vert\vert_{1}.
\end{eqnarray*}
We have approximated the conditional distribution as $\mathbb{P}(u[k] \big\vert z[k], x[k]; A^{(l)}) \approx \mathds{1}_{\left\{u[k] = \overset{\ast}{u}[k]\right\}}$. In the final step of expectation, we can write
\begin{align}
&Q(A;A^{(l)}) =\mathbb{E}_{A^{(l)}}\left[\text{log}\mathbb{P}_{c}(z[k],x[k],u[k]) \big\vert x[k],z[k]\right]\nonumber\\
&\begin{aligned}
\qquad=\mathbb{E}_{u[k]|z[k],x[k];A^{(l)}}\left[\text{log}\mathbb{P}\left(z[k],x[k],u[k]; A\right)\right]\nonumber
\end{aligned}\\
 &\qquad= \begin{aligned}[t]
 &\mathbb{E}_{u[k]|z[k],x[k];A^{(l)}}\Bigl[\text{log}\mathbb{P}\Bigl(z[k],x[k]\big\vert u[k]; A\Bigr)\Bigr]\nonumber \\
&{+}\:\text{log}\mathbb{P}\left(u[k]\right)\mathds{1}_{\left\{u[k] = \overset{\ast}{u}[k]\right\}} \nonumber
\end{aligned}\\
&\qquad= \begin{aligned}[t]
\text{log}\mathbb{P}\left(z[k],x[k]\big\vert u[k]; A\right)\mathds{1}_{\left\{u[k] = \overset{\ast}{u}[k]\right\}},\nonumber
\end{aligned}
\end{align}
where $\mathbb{P}_{c}$ is used to signify the likelihood of the complete data, and constants are ignored. For the Maximization step, we can simply write
\begin{eqnarray*}
A^{(l+1)} &=& \text{arg}\max\limits_{A}Q(A;A^{(l)}) \\
&=& \text{arg}\max\limits_{A}\text{log}\mathbb{P}\left(z[k],x[k]\big\vert u[k];A\right)\mathds{1}_{\left\{u[k] = \overset{\ast}{u}[k]\right\}},
\end{eqnarray*}
\noindent or in other words,
\begin{equation*}
a_{i}^{(l+1)} = \text{arg}\min\limits_{a} \vert\vert \tilde{Z}_{i} - Xa\vert\vert_{2}^{2},
\end{equation*}
where any $k$th component of $\tilde{Z}_{i}$ is $\tilde{Z}_{i,k} = z_{i}[k] - \overset{\ast}{u}[k]^{T}b_{i}$.

\section{Proof of Proposition\,\ref{prop:liklBounded}}
\label{appnd:liklBounded}

\begin{proof}

We show that the likelihood for incomplete (observed) data is bounded at each EM update step. Let us denote the likelihood of the observed data in relation to the parameter $A^{(l)}$ as
\begin{equation}
\mathbb{P}(A^{(l)}) = \mathbb{P}(z,x;A^{(l)}),
\end{equation}
\noindent which is further written as
\begin{eqnarray*}
\mathbb{P}(A^{(l)}) &\propto& \int \mathbb{P}(z,x\big\vert u;A^{(l)})\mathbb{P}(u)du \\
&=& C\int\text{exp}\left(-\frac{1}{2\sigma^{2}}\vert\vert z - A^{(l)}x - Bu\vert\vert_{2}^{2}\right)\\
&&\qquad\text{exp}\left(-\lambda\vert\vert u\vert\vert_{1}\right)du\\
&\leq& C\int\text{exp}\left(-\lambda\vert\vert u\vert\vert_{1}\right)du \leq \mathcal{O}(1),\\
\end{eqnarray*}
\noindent where $C$ is the normality constant. Therefore $\mathbb{P}(A^{(l)})$ is bounded for every iteration index $l\geq 0$.
\end{proof}

\section{Proof of Theorem\,\ref{thm:subMod}}
\label{appnd:subMod}
\begin{proof}

For a given $n\times m$ matrix $A$, let $\mathcal{R}(S)$ denote the span of rows of $A$ indexed by set $S$. Let $f(S)$ be the rank of matrix composed by rows indexed by set $S$, therefore $f(S) = \vert\mathcal{R}(S)\vert$. For given $\Omega = \{1,2,\hdots,n\}$, $S \subseteq T \subseteq \Omega$ and $a \in \Omega \setminus T$, we can write

\begin{eqnarray*}
\vert\mathcal{R}(T\cup a)\vert &=& \vert\mathcal{R}(S\cup a)\vert + \vert\mathcal{R}(T\setminus S) \cap \mathcal{R}(S \cup a)^{\perp}\vert \\
&=& \vert\mathcal{R}(S\cup a)\vert + \vert\mathcal{R}(T\setminus S) \cap \mathcal{R}(S)^{\perp} \cap \mathcal{R}(a)^{\perp}\vert\\
&\leq& \vert\mathcal{R}(S\cup a)\vert + \vert\mathcal{R}(T\setminus S) \cap \mathcal{R}(S)^{\perp}\vert\\
&=& \vert\mathcal{R}(S\cup a)\vert + \vert\mathcal{R}(T)\vert - \vert\mathcal{R}(S)\vert
\end{eqnarray*} 
\noindent where the last equality is written using the fact that dimension of intersection of $\mathcal{R}(T\setminus S)$ and orthogonal space of $\mathcal{R}(S)$ are number of linearly independent rows in $\mathcal{R}(T)$ which are not in $\mathcal{R}(S)$ i.e. $\vert\mathcal{R}(T)\vert - \vert\mathcal{R}(S)\vert$.
\end{proof}
}
}

%
%

\footnotesize

\bibliographystyle{IEEEtran}
\bibliography{IEEEabrv,senPlac}

\end{document}